\documentclass{article}
\RequirePackage{textcomp}
\RequirePackage{latexsym}
\RequirePackage{amstext}
\RequirePackage{amsmath}
\RequirePackage{amssymb}
\RequirePackage{amsthm}
\usepackage{fullpage}
\usepackage{hyperref}
\usepackage[boxed]{algorithm2e}
\usepackage{mdwlist}
\usepackage{float}
\usepackage{rotating}

\theoremstyle{plain}
\newtheorem{theorem}{Theorem}[section]
\newtheorem{lemma}[theorem]{Lemma}
\newtheorem{corollary}[theorem]{Corollary}

\def\ep{\varepsilon}

\def\DD{\mathcal{D}}

\def\m{\!-\!}
\def\pl{\!+\!}
\def\eq{\!=\!}

\newcommand{\p}[1]{(#1)}

\newcommand{\pc}[1]{\{#1\}}

\newcommand{\abs}[1]{|#1|}
\newcommand{\set}[1]{\pc{#1}}
\newcommand{\setc}[2]{\pc{#1 : #2}}

\newcommand{\expp}[1]{\exp\p{#1}}
\newcommand{\logp}[1]{\log\p{#1}}

\renewcommand{\P}[1]{P\p{#1}}
\newcommand{\E}[1]{E\p{#1}}

\newcommand{\OO}[1]{O\p{#1}}
\newcommand{\OM}[1]{\Omega\p{#1}}
\newcommand{\OT}[1]{\Omega\p{#1}}

\newcommand{\poly}[1]{\text{poly}\p{#1}}

\newcommand{\rank}[2]{R\p{#1,#2}}
\newcommand{\rhat}[2]{\hat{R}\p{#1,#2}}

\newcommand{\eva}[1]{\alpha_{#1}}
\newcommand{\evb}[1]{\beta_{#1}}

\title{A randomized online quantile summary in $\OO{\frac{1}{\ep} \log \frac{1}{\ep}}$ words}

\author{
David Felber
\thanks{University of California at Los Angeles. \texttt{dvfelber@cs.ucla.edu}.}
\and
Rafail Ostrovsky
\thanks{University of California at Los Angeles. \texttt{rafail@cs.ucla.edu}.}
}
\date{}

\begin{document}

\maketitle

\begin{abstract}

  A quantile summary is a data structure that approximates to $\ep$-relative
  error the order statistics of a much larger underlying dataset.

  In this paper we develop a randomized online quantile summary for the cash
  register data input model and comparison data domain model that uses
  $\OO{\frac{1}{\ep} \log \frac{1}{\ep}}$ words of memory. This improves upon
  the previous best upper bound of $\OO{\frac{1}{\ep} \log^{3/2} \frac{1}{\ep}}$
  by Agarwal et. al. (PODS 2012). Further, by a lower bound of Hung and Ting
  (FAW 2010) no deterministic summary for the comparison model can outperform
  our randomized summary in terms of space complexity. Lastly, our summary has
  the nice property that $\OO{\frac{1}{\ep} \log \frac{1}{\ep}}$ words suffice
  to ensure that the success probability is $1 - e^{-\poly{1/\ep}}$.

\end{abstract}

\section{Introduction}
\label{sec:introduction}

A quantile summary $S$ is a fundamental data structure that summarizes an
underlying dataset $X$ of size $n$, in space much less than $n$. Given a query
$\phi$, $S$ returns a sample $y$ of $X$ such that the rank of $y$ in $X$ is
(probably) approximately $\phi n$. Quantile summaries are used in sensor
networks to aggregate data in an energy-efficient manner and in database query
optimizers to generate query execution plans.

Quantile summaries have been developed for a variety of different models and
metrics. The data input model we consider is the standard online cash register
streaming model, in which a new item is added to the dataset at each new
timestep, and the total number of items is not known until the end. The data
domain model we consider is the comparison model, in which stream items come
from an arbitrary ordered domain (and specifically, not necessarily from the
integers).

Formally, our quantile summary problem is defined over a totally ordered domain
$\DD$ and by an error parameter $\ep \le 1/2$. There is a dataset $X$ that is
initially empty. Time occurs in discrete steps. In timestep $t$, stream item
$x_t$ arrives and is then processed, and then any quantile queries $\phi$ in
that step are received and processed. To be definite, we pick the first timestep
to be $1$. We write $X_t$ or $X(t)$ for the $t$-item prefix stream $x_1 \ldots
x_t$ of $X$. The goal is to maintain at all times $t$ a summary $S_t$ of the
dataset $X_t$ that, given any query $\phi$ in $(0, 1]$, can return a sample $y =
  y\p{\phi}$ so that $\abs{\rank{y}{X_t} - \phi t} \le \ep t$, where
  $\rank{a}{Z}$ is the \emph{rank of item $a$ in set $Z$}, defined as
  $\abs{\setc{z \in Z}{a \le z}}$. For randomized summaries, we only require
  that $\forall t \forall \phi,\; \P{\abs{\rank{y}{X_t} - \phi t} \le \ep t}
  \,\ge\, 2/3$; that is, $y$'s rank is only probably close to $\phi t$, not
  definitely close. In fact, it will be easier to deal with the rank directly,
  so we define $\rho = \phi t$ and use that in what follows.

\subsection{Previous work}

The two most directly relevant pieces of prior work are randomized online
quantile summaries for the cash register/comparison model. Aside from oblivious
sampling algorithms (which require storing $\OT{1/\ep^2}$ samples) we are
unaware of any other randomized online quantile summaries that work in the
comparison model.

The newer of the two is that of Agarwal, Cormode, Huang, Phillips, Wei, and Yi
\cite{ACHPWY2012} \cite{ACHPWY2013}. Among other results, Agarwal et. al.
develop a randomized online quantile summary for the cash register/comparison
model that uses $\OO{\frac{1}{\ep} \log^{3/2} \frac{1}{\ep}}$ words of memory.
This summary has the nice property that any two such summaries can be combined
to form a summary of the combined underlying dataset without loss of accuracy or
increase in size.

The earlier such summary is that of Manku, Rajagopalan, and Lindsay
\cite{MRL1999}, which uses $\OO{\frac{1}{\ep} \log^2 \frac{1}{\ep}}$ space. At a
high level, their algorithm downsamples the input stream in a non-uniform way
and feeds the downsampled stream into a deterministic summary, which
periodically adjusts the downsampling rate.

We note here that our algorithm is inspired by the algorithm of Manku et. al.
but has important differences. We defer a discussion of the similarities and
differences to section \ref{sec:discussion} after the presentation of our
algorithm in section \ref{sec:online}.

For the comparison model, the best deterministic online summary to date is the
(GK) summary of Greenwald and Khanna \cite{GK2001}, which uses
$\OO{\frac{1}{\ep} \log \ep n}$ space. This improved upon a deterministic (MRL)
summary of Manku, Rajagopalan, and Lindsay \cite{MRL1998} and a summary implied
by Munro and Paterson \cite{MP1978}, which use $\OO{\frac{1}{\ep} \log^2 \ep n}$
space.

A more restrictive domain model than the comparison model is the bounded
universe model, in which elements are drawn from the integers $\{1, \ldots,
u\}$. For this model there is a deterministic online summary by Shrivastava,
Buragohain, Agrawal, and Suri \cite{SBAS2004} that uses $\OO{\frac{\log
    u}{\ep}}$ space.

Not much exists in the way of lower bounds for this problem. There is a simple
lower bound of $\OM{1/\ep}$ which intuitively comes from the fact that no one
sample can satisfy more than $2 \ep n$ different rank queries. For the
comparison model, Hung and Ting \cite{HT2010} developed a deterministic
$\OM{\frac{1}{\ep} \log \frac{1}{\ep}}$ lower bound. Whether this bound can be
extended to hold for our weaker probabilistic guarantee, and whether our
algorithm can be modified to satisfy the stronger deterministic guarantee, are
both open questions.

\subsection{Our results}

In the next section we describe a simple $\OO{\frac{1}{\ep} \log \frac{1}{\ep}}$
streaming summary that is online except that it requires $n$ to be given up
front and that it is unable to process queries until it has seen a constant
fraction of the input stream. In section \ref{sec:online} we develop this simple
summary into a fully online summary that can answer queries at any point in
time. We close in section \ref{sec:discussion} by examining the similarities and
differences between our summary and previous work and discuss a design approach
for similar streaming problems.

\section{A simple streaming summary}
\label{sec:simple}

Before we describe our algorithm we must first describe its two main components
in a bit more detail than was used in the introduction. The two components are
Bernoulli sampling and the GK summary \cite{GK2001}.

\subsection{Bernoulli sampling}

Bernoulli sampling downsamples a stream $X$ of size $n$ to a sample stream $S$
by choosing to include each next item into $S$ with independent probability
$m/n$. (As stated this requires knowing the size of $X$ in advance.) At the end
of processing $X$, the expected size of $S$ is $m$, and the expected rank of any
sample $y$ in $S$ is $\E{\rank{y}{S}} = \frac{m}{n} \rank{y}{X}$.
In fact, for any times $t \le n$ and partial streams $X_t$ and $S_t$, where
$S_t$ is the sample stream of $X_t$, we have $\E{|S_t|} = m t / n$ and
$\E{\rank{y}{S_t}} = \frac{m}{n} \rank{y}{X_t}$. To generate an estimate for
$\rank{y}{X_t}$ from $S_t$ we use $\rhat{y}{X_t} = \frac{n}{m} \rank{y}{S_t}$.
The following theorem bounds the probability that $S$ is very large or that
$\rhat{y}{X_t}$ is very far from $\rank{y}{X_t}$ (for any given time $t \ge
n/64$, but not for all times $t = n/64 \ldots n$ combined). The proof is
folklore, a simple application of Chernoff bounds.

\begin{theorem}
  \label{thm:bernoulli}
  For all times $t \ge n/64$, $\P{\abs{S_t} > 2 t m / n} \,<\, \expp{-m/192}$.

  Further, for all times $t \ge n/64$ and items $y$, $\P{\abs{\rhat{y}{X_t} -
      \rank{y}{X_t}} > \ep t / 8} \,<\, 2 \expp{-\ep^2 m/12288}$.
\end{theorem}
\begin{proof}
  For the first part, $\P{\abs{S_t} > 2 t m / n} \,<\, \expp{-t m / 3 n} \,<\,
  \expp{-m / 192}$ (since $t \ge n/64$).

  For the second part, $\P{\abs{\rhat{y}{X_t} - \rank{y}{X_t}} > \ep t / 8}$ is
  equal to $\P{\abs{\rank{y}{S_t} - \E{\rank{y}{S_t}}} > \ep t m / 8 n}$. The
  Chernoff bound is $\P{\abs{\rank{y}{S_t} - \E{\rank{y}{S_t}}} > \delta
    \E{\rank{y}{S_t}}} \,<\, 2 \expp{-\min\set{\delta,\delta^2}
    \E{\rank{y}{S_t}} / 3}$. Here, $\delta = \ep t / 8 \rank{y}{S_t}$, so $P < 2
  \expp{-\ep^2 t^2 m / 192 n \E{\rank{y}{S_t}}} \le 2 \expp{-\ep^2 m / 12288}$.
\end{proof}

This means that, given any $1 \le \rho \le t$, if we return the sample $y \in
S_t$ with $\rank{y}{S_t} = \rho m / n$, then $\rank{y}{X_t}$ is likely to be
close to $\rho$.

\subsection{GK summary}

The GK summary is a deterministic summary that can answer queries to relative
error, over any portion of the received stream. If $G_t$ is the summary after
inserting the first $t$ items $X_t$ from stream $X$ into $G$ then, given any $1
\le \rho \le t$, $G_t$ can return a sample $y \in X_t$ so that
$\abs{\rank{y}{X_t} - \rho} \le \ep t / 8$. Greenwald and Khanna guarantee in
\cite{GK2001} that $G_t$ uses $\OO{\frac{1}{\ep} \logp{\ep t}}$ words. We call
this the \emph{GK guarantee}.

\subsection{Our summary}

We combine Bernoulli sampling with the GK summary by downsampling the input data
stream $X$ to a sample stream $S$ and then feeding $S$ into a GK summary $G$. It
looks like this:

\begin{figure}[H]
  \includegraphics{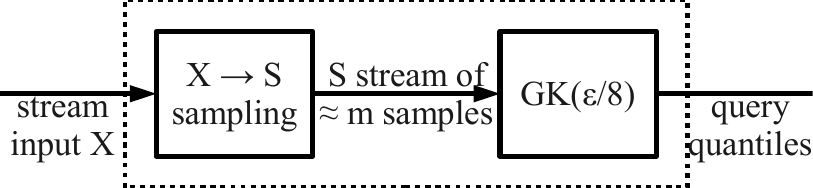}
  \centering
  \caption{The big picture.}
  \label{fig:figure1}
\end{figure}

The key reason this gives us a small summary is that we never need to store $S$;
each time we sample an item into $S$ we immediately feed it into $G$. Therefore,
we only use as much space as $G\p{S\p{X_t}}$ uses. In particular, as long as $m
= \OO{\poly{1/\ep}}$, we use only $\OO{\frac{1}{\ep} \log \frac{1}{\ep}}$ words.

To answer a query $\rho$ for $X_t$ we ask $G_t$ the query $\rho m / n$ and
return the resulting sample $y$. There is a slight issue in that $\rho m / n$
may be larger than $\abs{S}$; but if the approximation guarantee holds for the
largest item in $X_t$ then $\rho m / n < \p{t \pl \ep t/8} m / n$, so using
$\min\pc{\rho m / n, \abs{S}}$ instead will not cause more than $\ep/8$ relative
error in the approximation.

The probability that our sample stream $S_t$ is not too big (uses more than $2 t
m / n$ samples) is at least $1 - \expp{-m/192}$. If this happens to be the case
then the probability that all of its samples $y$ are good (have
$\abs{\rank{y}{S_t} - \E{\rank{y}{S_t}}} \le \ep t m / 8 n$) is at least $1 - 4
m \expp{- \ep^2 m / 12288}$ by theorem \ref{thm:bernoulli} and the union bound.
Choosing $m \ge \frac{300000 \ln 1/\ep}{\ep^2}$ suffices to guarantee that both
events occur with total probability at least $2/3$.

Further, if both $S_t$ events occur then the total error introduced by both
$S_t$ and $G_t$ is at most $\ep t / 2$. Suppose that $G_t$ returns $y$ when
given $\rho m / n$. This means that $\abs{\rank{y}{S_t} - \rho m / n} \le \ep
\abs{S_t} \le \ep \p{2 t m / n} / 8$ by the GK guarantee. Since both events for
$S_t$ occur, we also have $\abs{\rank{y}{S_t} - \frac{m}{n} \rank{y}{X_t}} \le
\ep t m / 4 n$ (and only $\ep t m / 8 n$ in the case that we don't truncate
$\rho m / n$ to $\abs{S}$). Thus, $\abs{\frac{m}{n} \rank{y}{X_t} - \rho m / n}
\le \ep t m / 2 n$. Equivalently, $\abs{\rank{y}{X_t} - \rho} \le \ep t / 2$.

\subsection{Caveats}

There are two serious issues with this summary. The first is that it requires us
to know the value of $n$ in advance to perform the sampling. Also, as a
byproduct of the sampling, we can only obtain approximation guarantees after we
have seen at least $1/64$ (or at least some constant fraction) of the items.
This means that while the algorithm is sufficient for approximating order
statistics over streams stored on disk, more is needed to get it to work for
online streaming applications, in which (1) the stream size $n$ is not known in
advance, and (2) queries can be answered approximately at all times $t \le n$
and not just when $t \ge n/64$.

Adapting the idea of our basic streaming summary to work online constitutes the
next section and the bulk of our contribution.
We start with a high-level overview of our online summary algorithm.
In section \ref{sec:online-alg} we formally define an initial version of our
algorithm whose expected size at any given time is $\OO{\frac{1}{\ep} \log
  \frac{1}{\ep}}$ words.
In section \ref{sec:online-error} we show that our algorithm gurantees that
$\forall n \forall \rho,\; \P{\abs{\rank{y}{X_n} - \rho} \le \ep n} \,\ge\, 1 -
\expp{-1/\ep}$.
In section \ref{sec:online-spacetime} we discuss the slight modifications
necessary to get a deterministic $\OO{\frac{1}{\ep} \log \frac{1}{\ep}}$ space
complexity, and also perform a time complexity analysis.

\section{An online summary}
\label{sec:online}

Our algorithm works in \emph{rows}, which are illustrated in appendix
\ref{sec:appendix-online-alg-diagram}. Row $r$ is a summary of the first $2^r 32
m$ stream items. Since we don't know how many items will actually be in the
stream, we can't start all of these rows running at the outset. Therefore, we
start each row $r \ge 1$ once we have seen $1/64$ of its total items. However,
since we can't save these items for every row we start, we need to construct an
approximation of this fraction of the stream, which we do by using the summary
of the previous row, and join this approximating stream with the new items that
arrive while the row is live. We then wait until the row has seen a full half of
its items before we permit it to start answering queries; this dilutes the
influence of approximating the $1/64$ of its input that we couldn't store.

Operation within a row is very much like the operation of our fixed-$n$
streaming summary. We feed the joint approximate prefix + new item stream
through a Bernoulli sampler to get a sample stream, which is then fed into a GK
summary (which is stored). After row $r$ has seen half of its items, its GK
summary becomes the one used to answer quantile queries. When row $r+1$ has seen
$1/64$ of \emph{its} total items, row $r$ generates an approximation of those
items from its GK summary and feeds them as a stream into row $r+1$.

Row $0$ is slightly different in order to bootstrap the algorithm. There is no
join step since there is no previous row to join. Also, row $0$ is active from
the start. Lastly, we get rid of the sampling step so that we can answer queries
over timesteps $1 \ldots m/2$.

After the first $32 m$ items, row $0$ is no longer needed, so we can clean up
the space used by its GK summary. Similarly, after the first $2^r 32 m$ items,
row $r$ is no longer needed. The upshot of this is that we never need storage
for more than six rows at a time. Since each GK summary uses $\OO{\frac{1}{\ep}
  \log \frac{1}{\ep}}$ words, the six live GK summaries use only a constant
factor more.

Our error analysis, on the other hand, will require us to look back as many as
$\OT{\log 1/\ep}$ rows to ensure our approximation guarantee. We stress that we
will not need to actually \emph{store} these $\OT{\log 1/\ep}$ rows for our
guarantee to hold; we will only need that they didn't have any bad events (as
will be defined) when they \emph{were} alive.

\subsection{Algorithm description}
\label{sec:online-alg}

Our algorithm works in rows. Each row $r$ has its own copy $G_r$ of the GK
algorithm that approximates its input to $\ep/8$ relative error. For each row
$r$ we define several streams: $A_r$ is the prefix stream of row $r$, $B_r$ is
its suffix stream, $R_r$ is its prefix stream replacement (generated by the
previous row), $J_r$ is the joint stream $R_r$ followed by $B_r$, $S_r$ is its
sample stream, and $Q_r$ is a one-time stream generated from $G_r$ by querying
it with ranks
$\rho_1 \ldots \rho_{8/\ep}$,
where $\rho_q = q \p{\ep/8} \p{m/64}$.

The prefix stream $A_r = X\p{2^{r-1} m}$ for row $r \ge 1$, importantly, is not
directly received by row $r$. Instead, at the end of timestep $2^{r-1} m$, row
$r \m 1$ generates $Q_{r-1}$ and duplicates each of those $8 / \ep$ items
$2^{r-1} \ep m / 8$ times to get the replacement prefix $R_r$, which is then
immediately fed into row $r$ before timestep $2^{r-1} m \pl 1$ begins.

Each row can be \emph{live} or not and \emph{active} or not. Row $0$ is live in
timesteps $1 \ldots 32 m$ and row $r \ge 1$ is live in timesteps $2^{r-1} m \pl
1 \ldots 2^r 32 m$. Live rows require space; once a row is no longer live we can
free up the space it used. Row $0$ is active in timesteps $1 \ldots 32 m$ and
row $r \ge 1$ is active in timesteps $2^r 16 m \pl 1 \ldots 2^r 32 m$. This
definition means that exactly one row $r\p{t}$ is active in any given timestep
$t$. Any queries that are asked in timestep $t$ are answered by $G_{r\p{t}}$.
Given query $\rho$, we ask $G_{r\p{t}}$ for $\rho/2^{r\p{t}} 32$ and return the
result.

At each timestep $t$, when item $x_t$ arrives, it is fed as the next item in the
suffix stream $B_r$ for each live row $r$. $B_r$ joined with $R_r$ defines the
joined input stream $J_r$. For $r \ge 1$, $J_r$ is downsampled to the sample
stream $S_r$ by sampling each item independently with probability $1/2^r 32$.
For row $0$, no downsampling is performed, so $S_0 = J_0$. Lastly, $S_r$ is fed
into $G_r$.

Appendix \ref{sec:appendix-online-alg-diagram} shows the operation of and the
communication between the first six rows. Solid arrows indicate continuous
streams and dashed arrows indicate one-time messages. Appendix
\ref{sec:appendix-online-alg-listing} is a pseudocode listing of the algorithm.

\subsection{Error analysis}
\label{sec:online-error}

Define $C_r = x\p{2^r 32 m \pl 1}, x\p{2^r 32 m \pl 2}, \ldots$ and $Y_r$ to be
$R_r$ followed by $B_r$ and then $C_r$. That is, $Y_r$ is just the continuation
of $J_r$ for the entire length of the input stream.

Fix some time $t$. All of our claims will be relative to time $t$; that is, if
we write $S_r$ we mean $S_r\p{t}$. Our error analysis proceeds as follows. We
start by proving that $\rank{y}{Y_r}$ is a good approximation of
$\rank{y}{Y_{r-1}}$ when certain conditions hold for $S_{r-1}$. By induction,
this means that $\rank{y}{Y_r}$ is a good approximation of $\rank{y}{X \eq Y_0}$
when the conditions hold for all of $S_0 \ldots S_{r-1}$, and actually it's
enough for the conditions to hold for just $S_{r - \log 1/\ep} \ldots S_{r-1}$
to get a good approximation. Having proven this claim, we then prove that the
result $y = y\p{\rho}$ of a query to our summary has $\rank{y}{X}$ close to
$\rho$. Lastly, we show that $m = \OO{\poly{1/\ep}}$ suffices to ensure that the
conditions hold for $S_{r - \log 1/\ep} \ldots S_{r-1}$ with very high
probability ($1 - e^{-1/\ep}$).

\begin{lemma}
  \label{lem:online-error-indstep}
  Let $\eva{r}$ be the event that $\abs{S_r} > 2 m$ and let $\evb{r}$ be the
  event that any of the first $\le 2 m$ samples $z$ in $S_r$ has $\abs{2^r 32
    \rank{z}{S_r} - \rank{z}{Y_r}} > \ep t / 8$. Say that $S_r$ is \emph{good}
  if neither $\eva{r}$ nor $\evb{r}$ occur (or if $r = 0$).

  For all $r \ge 1$ such that $t \ge t_r = 2^{r-1} m$, and for all items $y$, if
  $S_{r-1}$ is good then $\abs{\rank{y}{Y_r} - \rank{y}{Y_{r-1}}} \le 2^r \ep
  m$.
\end{lemma}
\begin{proof}
  At the end of time $t_r$ we have $Y_r\p{t_r} = R_r\p{t_r}$, which is each item
  $y\p{\rho_q}$ in $Q_{r-1}$ duplicated $\ep t_r / 8$ times. If $S_{r-1}\p{t_r}$
  is good then by theorem \ref{thm:bernoulli} and the GK guarantee we have that
  $\abs{\rank{y\p{\rho_q}}{Y_{r-1}\p{t_r}} - 2^{r-1} 32 \rho_q} \le \ep t_r / 2$.

  Fix $q$ so that $y\p{\rho_q} \le y < y\p{\rho_{q+1}}$, where $y\p{\rho_0}$ and
  $y\p{\rho_{1 \pl 8/\ep}}$ are defined to be $\min X_t$ and $\sup \DD$ for
  completeness. Fixing $q$ this way implies that $\rank{y}{Y_r\p{t_r}} = 2^{r-1}
  32 \rho_q$. By the above bound on $\rank{y\p{\rho_q}}{Y_{r-1}\p{t_r}}$ we also
  have that $2^{r-1} 32 \rho_q - \ep t_r / 2 \le \rank{y}{Y_{r-1}\p{t_r}} <
  2^{r-1} 32 \rho_{q+1} + \ep t_r / 2$.

  Putting these two bounds together, and recalling that $\rho_q = q \ep m /
  512$, we find that $\abs{\rank{y}{Y_r\p{t_r}} - \rank{y}{Y_{r-1}\p{t_r}}} \le
  2^r \ep m$. For each time $t$ after $t_r$, the new item $x_t$ changes the rank
  of $y$ in both streams $Y_r$ and $Y_{r-1}$ by the same additive offset, so
  $\abs{\rank{y}{Y_r} - \rank{y}{Y_{r-1}}} = \abs{\rank{y}{Y_r\p{t_r}} -
    \rank{y}{Y_{r-1}\p{t_r}}} \le 2^r \ep m$.
\end{proof}

By applying this lemma inductively we can bound the difference between $Y_r$ and
$X = Y_0$:

\begin{corollary}
  For all $r \ge 1$ such that $t \ge t_r = 2^{r-1} m$, if all of $S_0\p{t_1},
  S_1\p{t_2}, \ldots,$ $S_{r-1}\p{t_r}$ are good, then $\abs{\rank{y}{Y_r} -
    \rank{y}{X}} \le 2^{r+1} \ep m$.
\end{corollary}

To ensure that all of these $S_i$ are good would require $m$ to grow with $n$,
which would be bad. Happily, it is enough to require only the last $\log_2
1/\ep$ sample summaries to be good, since the other items we disregard
constitute only a small fraction of the total stream.

\begin{corollary}
  \label{cor:online-error-lastgood}
  Let $d = \log_2 1/\ep$. For all $r \ge 1$ such that $t \ge t_r = 2^{r-1} m$,
  if all of $S_{r-1}\p{t_r}, \ldots, S_{r-d}\p{t_{r-d+1}}$ are good, then
  $\abs{\rank{y}{Y_r} - \rank{y}{X}} \le 2^{r+2} \ep m$.
\end{corollary}
\begin{proof}\let\qed\relax
  By lemma \ref{lem:online-error-indstep} we have $\abs{\rank{y}{Y_r} -
    \rank{y}{Y_{r-d}}} \le 2^{r+1} \ep m$. At time $t \ge t_{r-d}$, $Y_{r-d}$
  and $X$ share all except possibly the first $2^{\p{r-d}-1} m = 2^{r-1} m / 2^d
  = 2^{r-1} \ep m$ items. Thus
  \begin{align*}
    \abs{\rank{y}{Y_r} - \rank{y}{X}}
    &\le \abs{\rank{y}{Y_r} - \rank{y}{Y_{r-d}}} + \abs{\rank{y}{Y_{r-d}} - \rank{y}{X}} \\
    &\le 2^{r+1} \ep m + 2^r \ep m
    \\ & \omit\hfill\qedsymbol
  \end{align*}
\end{proof}

We now prove that the if the last several sample streams were good then querying
our summary will give us a good result.

\begin{lemma}
  \label{lem:online-error-querygood}
  Let $d = \log_2 \frac{1}{\ep}$ and $r = r\p{t}$. If all $S_r\p{t},
  S_{r-1}\p{t_r}, \ldots, S_{r-d}\p{t_{r-d+1}}$ are good, then querying our
  summary with rank $\rho$ (= querying the active GK summary $G_r$ with $\rho /
  2^r 32$) returns $y = y\p{\rho}$ such that $\abs{\rank{y}{X} - \rho} \le \ep
  t$.
\end{lemma}
\begin{proof}
  By corollary \ref{cor:online-error-lastgood} we have $\abs{\rank{y}{Y_r} -
    \rank{y}{X}} \le 2^{r+2} \ep m \le \ep t / 2$. By theorem
  \ref{thm:bernoulli} and the GK guarantee, $\abs{\rank{y}{Y_r} - \rho} \le \ep
  t / 2$.
\end{proof}

Lastly, we prove that $m = \OO{\poly{1/\ep}}$ suffices to ensure that all of
$S_r\p{t},$ $S_{r-1}\p{t_r}, \ldots, S_{r-d}\p{t_{r-d+1}}$ are good with
probability at least $1 - e^{-1/\ep}$.

\begin{lemma}
  \label{lem:online-error-goodprob}
  Let $d = \log_2 1/\ep$ and $r = r\p{t}$. If $m \ge \frac{400000 \ln
    1/\ep}{\ep^2}$ then all of $S_r\p{t}, S_{r-1}\p{t_r}, \ldots,
  S_{r-d}\p{t_{r-d+1}}$ are good with probability at least $1 - e^{-1/\ep}$.
\end{lemma}
\begin{proof}
  There are at most $1 \pl \log_2 1/\ep \le 4 \ln 1/\ep$ of these summary
  streams total. Theorem \ref{thm:bernoulli} and the union bound give us
  $\P{\text{no $\eva{r}$ occurs}} \le 4 \ln \frac{1}{\ep} \expp{-m/192}$ and
  $\P{\text{no $\evb{r}$ occurs}} \le 16 m \ln \frac{1}{\ep} \expp{-\ep^2
    m/12288}$.

  Together, $P = \P{\text{some $S_r$ is not good}} \le 20 m \ln \frac{1}{\ep}
  \expp{-\ep^2 m/12288}$. It suffices to choose $m \ge \frac{400000 \ln
    1/\ep}{\ep^2}$ to obtain $P \le e^{-1/\ep}$.
\end{proof}

\subsection{Space and time complexity}
\label{sec:online-spacetime}

A minor issue with the algorithm is that, as written in section
\ref{sec:online-alg}, we do not actually have a bound on the worst-case space
complexity of the algorithm; we only have a bound on the space needed at any
given point in time. This issue is due to the fact that there are low
probability events in which $\abs{S_r}$ can get arbitrarily large and the fact
that over $n$ items there are a total of $\OT{\log n}$ sample streams. The space
complexity of the algorithm is $\OO{\max \abs{S_r}}$, and to bound this value
with constant probability using the Chernoff bound appears to require that $\max
\abs{S_r} = \OM{\log \log n}$, which is too big.

Fortunately, fixing this problem is simple. Instead of feeding every sample of
$S_r$ into the GK summary $G_r$, we only feed each next sample if $G_r$ has seen
$< 2 m$ samples so far. That is, we deterministically restrict $G_r$ to
receiving only $2 m$ samples. Lemmas \ref{lem:online-error-indstep} through
\ref{lem:online-error-querygood} condition on the goodness of the sample streams
$S_r$, which ensures that the $G_r$ receive at most $2 m$ samples each, and the
claim of lemma \ref{lem:online-error-goodprob} is independent of the operation
of $G_r$. Therefore, by restricting each $G_r$ to receive at most $2 m$ inputs
we can ensure that the space complexity is deterministically $\OO{\frac{1}{\ep}
  \log \frac{1}{\ep}}$ without breaking our error guarantees.

From a practical perspective, the assumption in the streaming setting is that
new items arrive over the input stream $X$ at a high rate, so both the
worst-case per-item processing time as well as the amortized time to process $n$
items are important. For our per-item time complexity, the limiting factor is
the duplication step that occurs at the end of each time $t_r = 2^{r-1} m$,
which makes the worst-case per-item processing time as large as $\OT{n}$.
%
%
Instead, at time $t_r$ we could generate $Q_{r-1}$ and store it in $\OO{1/\ep}$
words, and then on each arrival $t = 2^{r-1} m \pl 1 \ldots 2^r m$ we could
insert both $x_t$ and also the next item in $R_r$. By the time $t_{r+1} = 2 t_r$
that we generate $Q_r$, all items in $R_r$ will have been inserted into $J_r$.
%
%
%
Thus the worst-case per-item time complexity is $\OO{\frac{1}{\ep}
  T_{\text{GK}}^{\text{max}}}$, where $T_{\text{GK}}^{\text{max}}$ is the
worst-case per-item time to query or insert into one of our GK summaries.
Over $2^r 32 m$ items there are at most $2 m$ insertions into any one GK
summary, so the amortized time over $n$ items in either case is $\OO{\frac{m
    \log n / 32 m}{n} T_{\text{GK}}}$, where $T_{\text{GK}}$ is the amortized
per-item time to query or insert into one of our GK summaries.

The pseudocode listing in appendix \ref{sec:appendix-online-alg-listing}
includes the changes of this section.

\section{Discussion}
\label{sec:discussion}

Our starting point is a very natural idea of Manku et. al. \cite{MRL1999} that
due to subtle technical difficulties saw no further application to the quantiles
problem for sixteen years. This key idea is to downsample the input stream and
feed the resulting sample stream into a deterministic summary data structure
(compare our figure \ref{fig:figure1} with figure 1 on page 254 of
\cite{MRL1999}). At a very high level, we are simply replacing their
deterministic $\OO{\frac{1}{\ep} \log^2 \ep n}$ MRL summary \cite{MRL1998} with
the deterministic $\OO{\frac{1}{\ep} \log \ep n}$ GK summary \cite{GK2001}.
However, as evidenced by the fact that fourteen years after the GK summary was
published the state of the art was the randomized $\OO{\frac{1}{\ep} \log^{3/2}
  \frac{1}{\ep}}$ summary of Agarwal et. al. \cite{ACHPWY2012}
\cite{ACHPWY2013}, adapting this idea to the GK summary without superconstant
overhead is nontrivial.

Our implementation of this idea is conceptually different from the
implementation of Manku et. al. in two respects. First, we use the GK algorithm
strictly as a black box, whereas Manku et. al. peek into the internals of their
MRL algorithm, using its algorithm-specific interface (\textsc{New},
\textsc{Collapse}, \textsc{Output}) rather than the more generic interface
(\textsc{Insert}, \textsc{Query}). At an equivalent level, dealing with the GK
algorithm is already unpleasant. Using the generic interface, our implementation
could just as easily replace the GK boxes in the diagram in appendix
\ref{sec:appendix-online-alg-diagram} with MRL boxes; or, for the bounded
universe model, with boxes running the q-digest summary of Shrivastava et. al.
\cite{SBAS2004}.

The second respect in which our algorithm differs critically from that of Manku
et. al. is that we operate on \emph{streams} rather than on stream \emph{items}.
We use this approach in our proof strategy too; the key step in our error
analysis, lemma \ref{lem:online-error-indstep}, is a statement about (what to us
are) static objects, so we can trade out the complexity of dealing with
time-varying data structures for a simple induction.

The approach we developed to reduce a deterministic summary to a randomized
summary was:
\begin{enumerate*}
\item For a fixed $n$, downsample the input stream, feed the resulting sample
  stream into the deterministic summary, and prove a probabilistic bound.
\item Run an infinite number of copies of step 1, for exponentially growing
  values of $n$.
\item Replace a constant fraction prefix of each copy with an approximation
  generated by the previous copy, and prove using step 1 that this approximation
  probably doesn't cause too much error.
\item Use step 3 inductively to prove a probabilistic bound for the entire
  stream.
\end{enumerate*}
We believe (albeit on the basis of this problem and our algorithm alone) that
developing streaming algorithms that operate on streams rather than on stream
items is likely to be a useful design approach for many problems.

\bibliographystyle{plain}
\bibliography{quantiles}


\appendix

\newpage
\section{Diagram for online algorithm}
\label{sec:appendix-online-alg-diagram}

\begin{figure}[H]
    \includegraphics[angle=90,scale=.91]{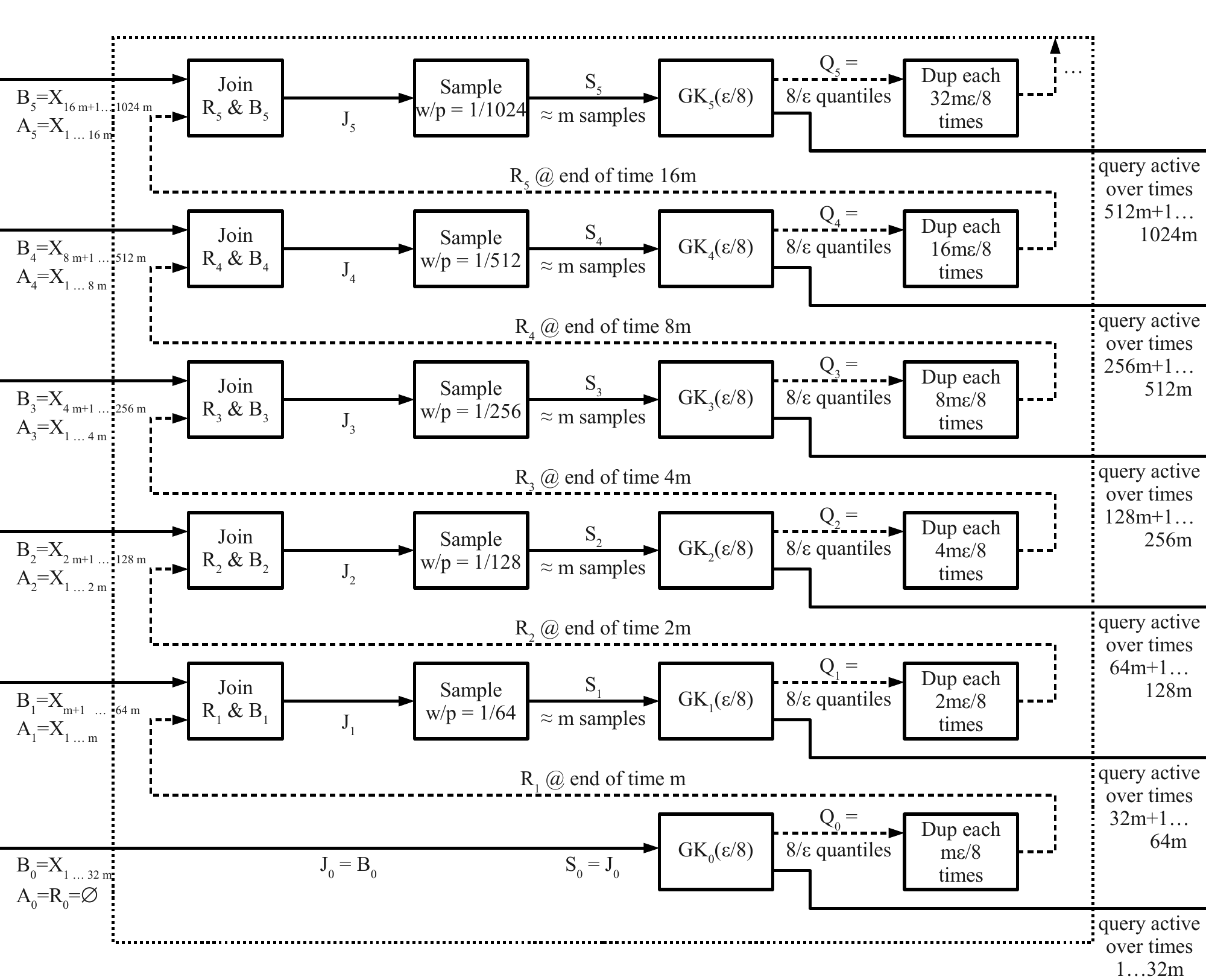}
    \centering
    \caption{Each row $r$ has its own copy $G_r$ of the GK algorithm that
      approximates its input to $\ep/8$ relative error. $A_r$ is the prefix
      stream of row $r$, $B_r$ is its suffix stream, $R_r$ is its prefix stream
      replacement (generated by the previous row), $J_r$ is the joint stream
      $R_r$ followed by $B_r$, $S_r$ is its sample stream, and $Q_r$ is a
      one-time stream generated from $G_r$ at time $2^r m$ to get the
      replacement prefix $R_{r+1}$.}
    \label{fig:figure2}
\end{figure}

\newpage
\section{Pseudocode for online algorithm}
\label{sec:appendix-online-alg-listing}

The differences in the algorithms of sections \ref{sec:online-alg} and
\ref{sec:online-spacetime} are marked.

\begin{algorithm}[H]
  \DontPrintSemicolon
  \SetAlgoNoEnd
  \SetAlgoNoLine
  \LinesNumbered
  \SetKwIF{On}{}{}{on}{do}{}{}{}
  \SetKwIF{Wp}{}{}{with probability}{do}{}{}{}
  %
  Initially, allocate space for $G_0$. Mark row $0$ as live and active.\;
  \For{$t = 1, 2, \ldots$}{
    \ForEach{live row $r \ge 0$}{
      \Wp{$1/2^r 32$}{
        \If{section \ref{sec:online-alg}}{
          Insert $x_t$ into $G_r$.\;
        }
        \ElseIf{section \ref{sec:online-spacetime}}{
          Insert $x_t$ into $G_r$ if $G_r$ has seen $< 2m$ insertions.\;
          \If{$r \ge 1$ and $2^{r-1} m < t \le 2^r m$ and $G_r$ has seen $< 2m$
            insertions}{
            \Wp{$1/2^r 32$}{
              Also insert item $t \m 2^{r-1} m$ of $R_r$ into $G_r$.\;
            }
          }
        }
      }
    }
    \If{$t = 2^{r-1} m$ for some $r \ge 1$}{
      Allocate space for $G_r$. Mark row $r$ as live.\;
      \If{section \ref{sec:online-alg}}{
        Query $G_{r-1}$ with $\rho_1 \ldots \rho_{8/\ep}$ to get $y_1 \ldots y_{8/\ep}$.\;
        \For{$q = 1 \ldots 8/\ep$}{
          \For{$1 \ldots 2^{r-1} \ep m/8$}{
            \Wp{$1/2^r 32$}{
              Insert $y_q$ into $G_r$.\;
            }
          }
        }
      }
      \ElseIf{section \ref{sec:online-spacetime}}{
        Store $Q_{r-1}$, to implicitly define $R_r$.\;
      }
    }
    \If{$t = 2^r 16 m$ for some $r \ge 1$}{
      Mark row $r$ as active. Unmark row $r \m 1$ as active.\;
    }
    \If{$t = 2^r 32 m$ for some $r \ge 0$}{
      Unmark row $r$ as live. Free space for $G_r$.\;
    }
  }
  \On{query $\rho$}{
    Let $r\p{t}$ be the active row.\;
    Query $G_{r\p{t}}$ for rank $\rho / 2^{r\p{t}} 32$. Return the result.\;
  }
  \caption{Procedural listing of the algorithm.}
  \label{fig:figure3}
\end{algorithm}

\end{document}